\documentclass[runningheads]{llncs}
\usepackage{amssymb}
\usepackage{amsmath}
\usepackage{hyperref}
\usepackage{cleveref}
\usepackage{cite}
\usepackage{verbatim}
\usepackage[T1]{fontenc}
\let\e\varepsilon
\newcommand{\D}{\mathbb{D}}
\newcommand{\Z}{\mathbb{Z}}
\def\set#1{\left\{ #1 \right\}}

\def\polylog{\operatorname{polylog}}

\let\Z\Integer

\def\set#1{\left\{ #1 \right\}}
\def\floor#1{\left\lfloor #1 \right\rfloor}

\def\abs#1{\left| #1 \right|}

\begin{document}
\title{Almost optimum $\ell$-covering of $\Z_n$} 

\author{Ke Shi\inst{1,2} \and
Chao Xu\inst{1}
}
\institute{
    School of Computer Science and Engineering,\\ University of Electronic Science and Technology of China
        \and
    School of Computer Science and Technology,\\ University of Science and Technology of China}

\maketitle

\begin{abstract}

A subset $B$ of the ring $\Z_n$ is referred to as a $\ell$-covering set if $\{ ab \pmod n \mid 0\leq a \leq \ell, b\in B\} = \Z_n$.
We show that there exists a $\ell$-covering set of $\Z_n$ of size $O(\frac{n}{\ell}\log n)$ for all $n$ and $\ell$, and how to construct such a set.
We also provide examples where any $\ell$-covering set must have a size of $\Omega(\frac{n}{\ell}\frac{\log n}{\log \log n})$. The proof employs a refined bound for the relative totient function obtained through sieve theory and the existence of a large divisor with a linear divisor sum. The result can be used to simplify a modular subset sum algorithm. 
\end{abstract}

\section{Introduction}
\label{sec:Intro}
For two sets $A,B\subseteq \Z_n$, we let $A\cdot B = \set{ab \pmod n \mid a\in A,b\in B}$. Let $[\ell] = \set{0,\ldots,\ell}$ be the natural numbers no larger than $\ell$.
A subset $B$ of the ring $\Z_n$ is termed a \emph{$\ell$-covering set} if $(\Z_n\cap [\ell])\cdot B=\Z_n$. Let $f(n,\ell)$ be the size of the smallest $\ell$-covering set of $\Z_n$, we are interested in finding $f(n,\ell)$. Equivalently, we can define a \emph{segment} of slope $i$ and length $\ell$ to be $\set{ ix \pmod n \mid x\in \Z_n\cap [\ell] }$, and we are interested in finding a set of segments that covers $\Z_n$.

$\ell$-coverings were used for flash storage related problems, including covering codes \cite{6151153, Klove.S2014, Klove.2016}, rewriting schemes\cite{6476065}. It also has been generalized to $\Z^d_n$ \cite{Klove.2016}. An $\ell$-covering is also useful in algorithm design. Since we can \emph{compress} a segment by dividing everything by its slope, an algorithm, where the running time depends on the size of the numbers in the input, can be improved. An implicit but involved application of $\ell$-covering was crucial for the first significant improvement to the modular subset sum problem \cite{Koili.X2019}.

The major question lies in finding the appropriate bound for $f(n,\ell)$. The trivial lower bound is $f(n,\ell) \geq \frac{n}{\ell}$. On the upper bound of $f(n,\ell)$, there are multiple studies where $\ell$ is a small constant, or $n$ has lots of structure, like being a prime number or maintaining certain divisibility conditions \cite{6151153, Klove.S2014, Klove.2016}. A fully general non-trivial upper bound for all $\ell$ and $n$ was first established by Chen et.al., which shows an explicit construction of an $O(\frac{n (\log n)^{\omega(n)}}{\ell^{1/2}})$ size $\ell$-covering set. They also showed $f(n,\ell) \leq \frac{n^{1+o(1)}}{\ell^{1/2}}$ using the fourth moment of character sums, but without providing a construction \cite{Chen.S.W2013}. In the same article, the authors show $f(p,\ell) = O(\frac{p}{\ell})$ for prime $p$ with an explicit construction. Koiliaris and Xu improved the result by a factor of $\sqrt{\ell}$ for general $n$ and $\ell$ using basic number theory, and showed $f(n,\ell) = \frac{n^{1+o(1)}}{\ell}$ \cite{Koili.X2019}. An $\ell$-covering set of equivalent size can also be found in $O(n\ell)$ time. The value hidden in $o(1)$ could be as large as $\Omega(\frac{1}{\log \log n})$, so it is relatively far from the lower bound. However, a closer examination of their result reveals that $f(n,\ell) = O(\frac{n}{\ell}\log n\log \log n)$ if $\ell$ is neither too large nor too small. That is, if $t \leq \ell \leq n/t$, where $t=n^{\Omega(\frac{1}{\log \log n})}$. See \Cref{fig:comp} for comparison of the results.

The covering problem can be considered in a more general context. For any \emph{semigroup} $(M, \diamond)$, define $A \diamond B = \set{a \diamond b \mid a\in A, b\in B}$. For $A\subseteq M$, we are interested in finding a small $B$ such that $A \diamond B = M$. Here $B$ is called an $A$-covering. The $\ell$-covering problem is the special case where the semigroup is $(\Z_n, \cdot)$, and $A=\Z_n\cap [\ell]$. When $M$ is a group, it was studied in \cite{Bollobas2011}. In particular, they showed for a finite group $(G,\diamond)$ and any $A\subseteq G$, there exists an $A$-covering of size no larger than $\frac{|G|}{|A|}(\log |A|+1)$. We wish to emphasize that our problem is based on the \emph{semigroup} $(\Z_n, \cdot)$, which is \emph{not a group}, and therefore, can exhibit very different behaviors. For example, if $A$ consists of only elements divisible by $2$ and $n$ is divisible by $2$, then no $A$-covering of $(\Z_n,\cdot)$ exists. It was shown that there exists $A$ that is a set of $\ell$ consecutive integers, any $A$-covering of $(\Z_n,\cdot)$ has $\Omega(\frac{n}{\ell}\log n)$ size \cite{roche2018packing}. Hence, the choice of the set $\Z_n\cap [\ell]$ is very special, as there are examples where $\ell$-covering has $O(\frac{n}{\ell})$ size \cite{Chen.S.W2013}. For reasons apparent in later part of the paper, we use $\ell$-covering in a semigroup $(X,\cdot)$ to mean a $(X\cap [\ell])$-covering. In the pursuit of our main theorem, another instance of the covering problem emerges, which might be of independent interest. Let the semigroup be $(\D_n,\odot)$, where $\D_n$ is the set of divisors of $n$, and $a \odot b = \gcd(ab, n)$, where $\gcd$ is the greatest common divisor function. We are interested in finding a $s$-covering set of $\D_n$ for some $s<n$.
\begin{figure*}
\begin{center}
\begin{tabular}{ |c|c|c| }
 \hline
  & Size of $\ell$-covering & Construction Time \\ 
 \hline
  & & \\[-1em]
 Chen et. al.    \cite{Chen.S.W2013}          & $O\left(\frac{n (\log n)^{\omega(n)}}{\ell^{1/2}}\right)$ & $\tilde{O}\left(\frac{n (\log n)^{\omega(n)}}{\ell^{1/2}}\right)$ \\ 
 & & \\[-1em]
 \hline
  & & \\[-1em]
 Chen et. al. \cite{Chen.S.W2013}         & $\frac{n^{1+o(1)}}{\ell^{1/2}}$ & Non-constructive \\ 
  & & \\[-1em]
 \hline
  & & \\[-1em]
 Koiliaris and Xu \cite{Koili.X2019} & $\frac{n^{1+o(1)}}{\ell}$ & $O(n\ell)$ \\ 
  & & \\[-1em]
  \hline 
  & & \\[-1em]
 \Cref{thm:main}       & $O(\frac{n}{\ell}\log n)$ & $O(n\ell)$\\
  \hline
 & & \\[-1em]
 \Cref{thm:randconstruction}       & $O(\frac{n}{\ell}\log n\log\log n)$ & $\tilde{O}(\frac{n}{\ell}) + n^{o(1)}$ randomized\\[+0.2em]
 \hline
\end{tabular}
\end{center}
\caption{Comparison of results for $\ell$-covering for arbitrary $n$ and $\ell$. $\omega(n)$ is the number of distinct prime factors of $n$.}
\label{fig:comp}
\end{figure*}

\subsection{Our Contributions}

\begin{enumerate}
\item We demonstrate that $f(n,\ell) = O(\frac{n}{\ell}\log n)$, and a slightly larger $\ell$-covering of size $O(\frac{n}{\ell}\log n \log \log n)$ can be found in $\tilde{O}(\frac{n}{\ell})+n^{o(1)}$ time.
\item We establish the existence of a constant $c>0$ and an infinite number of $n$ and $\ell$ pairs, such that $f(n,\ell) \geq c \frac{n}{\ell} \frac{\log n}{\log \log n}$.
\end{enumerate}

As an application, we show the new result simplifies the algorithm of \cite{Koili.X2019} for modular subset sums. In addition to these main contributions, we also offer some intriguing auxiliary results in number theory. These include a more precise bound for the relative totient function, as well as the discovery of a large divisor accompanied by a linear divisor sum.

\subsection{Technical overview}

Our approach is similar to the one of Koiliaris and Xu \cite{Koili.X2019}. We briefly describe their approach.
Recall $\Z_n$ is the set of integers modulo $n$. We further define $\Z_{n,d} =\set{x \mid \gcd(x,n)=d, x\in \Z_n}$, and $\Z^*_n = \Z_{n,1}$. Let $\mathcal{S}_\ell(X)$ be the set of segments of length $\ell$ and slope in $X$. Their main idea is to convert the covering problem over the \emph{semigroup} $(\Z_n,\cdot)$ to covering problems over the \emph{group} $(\Z^*_{n/d},\cdot)$ for all $d\in \D_n$. Since $\Z_{n,d}$ forms a partition of $\Z_n$, one can reason about covering them individually. That is, covering $\Z_{n,d}$ by $\mathcal{S}_\ell(\Z_{n,d})$. This is equivalent to covering $\Z^*_{n/d}$ with $\mathcal{S}_\ell(\Z^*_{n/d})$, and then lifting to a cover in $\Z_{n,d}$ by multiplying everything by $d$. Hence, now we only have to work with covering problems over $(\Z^*_{n/d},\cdot)$ for all $d$ and $n\geq 2$, all of which are \emph{groups}. The covering results for groups can be readily applied \cite{Bollobas2011}. Once we find the covering for each individual $(\Z^*_{n/d},\cdot)$, we take their union, and obtain an $\ell$-covering.

The approach was sufficient to obtain $f(n,\ell) = O(\frac{n}{\ell}\log n\log \log n)$ if $\ell$ is neither \emph{too small} nor \emph{too large}. However, their result suffers when $\ell$ is extreme in one of the two ways.

\begin{enumerate}
\item $\ell=n^{1-o(\frac{1}{\log \log n})}$: Any covering obtained would have size at least the number of divisors of $n$, which in the worst case can be $n^{\Omega(\frac{1}{\log \log n})}$, and dominates $\frac{n}{\ell}$.
\item $\ell=n^{o(\frac{1}{\log \log n})}$: If we are working on covering $\Z^*_n$, we need to know $|\Z^*_n\cap [\ell]|$, also known as $\phi(n,\ell)$. Previously, the estimate for $\phi(n,\ell)$ was insufficient when $\ell$ is small.
\end{enumerate}

Our approach can extend the applicable range to all $\ell$, and also eliminates the extra $\log \log n$ factor. There are two steps: First, we improve the estimate for $\phi(n,\ell)$. This improvement alone is sufficient to handle the cases when $\ell$ is relatively small compared to $n$. Second, we show that, roughly, a small $\ell'$-covering of $\D_n$ with some additional nice properties implies a small $\ell$-covering of $\Z_n$, where $\ell'$ is some number not too small compared to $\ell$. This change can shave off the $\log \log n$ factor.

\paragraph*{Organization}

The paper is organized as follows. \Cref{sec:prelim} contains the necessary number theory background. \Cref{sec:numbertheory} describes some number theoretical results on bounding $\phi(n,\ell)$, finding a large divisor of $n$ with a linear divisor sum, and covering of $\D_n$. \Cref{sec:bounds} proves the main theorem that $f(n,\ell) = O(\frac{n}{\ell}\log n)$, discusses its construction, and also provides a lower bound.

\section{Preliminaries}\label{sec:prelim}
This paper utilizes a few simple algorithmic concepts, but our methods are primarily analytical. Therefore, we have reserved some space in the preliminaries to set the scene.

Let $\mathcal{X}$ be a collection of subsets in some universe set $U$. A \emph{set cover} of $U$ is a subcollection of $\mathcal{X}$ whose union covers $U$. Formally, $\mathcal{X}'$ is a set cover of $U$ if $\mathcal{X}'\subseteq \mathcal{X}$ such that $U= \bigcup_{X\in \mathcal{X}'} X$. The \emph{set cover problem} is the computational problem of finding a set cover of minimum cardinality.

All multiplications in $\Z_n$ are modulo $n$, and henceforth we will omit the "$\pmod n$" notation.
A set of the form $\set{ ix \mid x \in \Z_n\cap [\ell]}$ is called a \emph{segment} of length $\ell$ with slope $i$. Note that a segment of length $\ell$ might contain fewer than $\ell$ elements. Recall that $\mathcal{S}_\ell(X)$ represents the collection of segments of length $\ell$ with slopes in $X$, namely $\set{ \set{ ix \mid x\in \Z_n\cap [\ell]} \mid i\in X}$. Thus, finding an $\ell$-covering is equivalent to the set cover problem where the universe is $\Z_n$ and the collection of subsets is $\mathcal{S}_\ell(\Z_n)$.

There are well-known bounds relating the size of a set cover to the frequency of each element in the cover.

\begin{theorem}[\cite{Lovas.1975,Stein.1974}]\label{thm:bettergreedy}
Let there be a collection of $t$ sets each with size at most $a$, and each element of the universe is covered by at least $b$ of the sets, then there exists a subcollection of $O(\frac{t}{b}\log a)$ sets that covers the universe.
\end{theorem}

The above theorem serves as our primary combinatorial tool for bounding the size of a set cover. To achieve a cover of the desired size, we find the greedy algorithm to be sufficient. It is worth noting that the group covering theorem for finite groups, as presented in \cite{Bollobas2011}, is a direct application of this principle.

In this context, the base of the $\log$ is $e$. To avoid dealing with negative values, we define $\log(x)$ as $\max(\log(x),1)$. We use $\tilde{O}(f(n))$, the soft $O$, as shorthand for $O(f(n)\polylog n)$.

\subsection{Number theory}
We utilize some standard notations and bounds listed in \Cref{fig:notations}, which can be found in various analytic number theory textbooks, for example, \cite{Davenport2000-nd}.

\begin{figure*}
\begin{center}
\begin{tabular}{|c|c|}
\hline
Notation & Definition/Property \\
\hline
$\Z_n$ & Set of integers modulo $n$ \\
$\Z_{n,d}$ & $\{x \mid \gcd(x,n)=d, x\in \Z_n\}$ \\
$\Z^*_n$ & Set of numbers in $\Z_n$ that are relatively prime to $n$ \\
$m|n$ & $m$ is a divisor of $n$ \\
$\pi(n)$ & Prime counting function, $\pi(n)=\Theta(\frac{n}{\log n})$ \\
$\phi(n)$ & Euler totient function, $\phi(n) = |\Z^*_n| = n\prod_{p|n,p\text{ prime}}\left(1-\frac{1}{p}\right)$ \\
$\phi(n)$ bound & $\phi(n) = \Omega(\frac{n}{\log \log n})$ \\
$\omega(n)$ & Number of distinct prime factors of $n$, $\omega(n) = O(\frac{\log n}{\log \log n})$ \\
$d(n)$ & Divisor function, $d(n) = n^{O(\frac{1}{\log \log n})}=n^{o(1)}$ \\
$\sigma(n)$ & Divisor sum function, $\sigma(n) \leq \frac{n^2}{\phi(n)}$, $\sigma(n) = O(n\log \log n)$ \\
$\sum_{p\leq n, p \text{ prime}}\frac{1}{p}$ & Sum of reciprocal of primes no larger than $n$, $O(\log \log n)$ \\
\hline
\end{tabular}
\caption{List of number theoretical notations and properties.} \label{fig:notations}
\end{center}
\end{figure*}

Our argument is centered around the \emph{relative totient function}, denoted as $\phi(n,\ell) = |\Z^*_n\cap [\ell]|$.

\begin{theorem}\label{thm:count}
Consider integers $0\leq \ell<n$, $y\in \Z_{n,d}$. The number of solutions $x\in \Z^*_n$ such that $xb\equiv y \pmod n$ for some $b\leq \ell$ is 
  \[
  \frac{\varphi(\frac{n}{d}, \floor{\frac{\ell}{d}})}{\varphi(\frac{n}{d})} \varphi(n).
  \]
\end{theorem}
\begin{proof}
See \Cref{sec:nproof}.
\end{proof}

\section{Number theoretical results}\label{sec:numbertheory}
This section we show some number theoretical bounds. The results are technical. The reader can skip the proofs of this section on first view. 

\subsection{Estimate for relative totient function}
This section proves a good estimate of $\phi(n,\ell)$ using sieve theory, the direction was hinted in \cite{252852}.

\begin{theorem}\label{thm:betterestimate}
There exists positive constant $c$, such that

\[
 \phi(n,\ell) = \begin{cases}
 \Omega(\frac{\ell}{n} \phi(n))& \text{ if } \ell > c \log^5 n\\
 \Omega(\frac{\ell}{\log \ell}) & \text{ if } \ell > c \log n\\
\end{cases}
\]

\end{theorem}

The proof can be found in \Cref{sec:betterestimate}. As a corollary, we prove a density theorem.
\begin{theorem}
\label{thm:betterunitcover}
There exists a constant $c$, such that for any $n$, and a divisor $d$ of $n$, if $\frac{\ell}{c \log^5 n} \geq d$, then each element in $\Z_{n,d}$ is covered $\Omega(\frac{n}{\ell}\phi(n))$ times by $\mathcal{S}_\ell(\Z^*_n)$. 
\end{theorem}

\begin{proof}
By \Cref{thm:count}, the number of segments in $\mathcal{S}_\ell(\Z^*_n)$ covering some fixed element in $\Z_{n,d}$ is $\frac{\phi(n/d,\ell/d)}{\phi(n/d)}\phi(n)$. As long as $\ell$ is not too small, $\phi(n,\ell) = \Omega(\frac{\ell}{n}\phi(n))$. In particular, by \Cref{thm:betterestimate}, if $\lfloor \ell/d\rfloor \geq c \log^5(n/d)$, we have $\phi(n/d,\ell/d)/\phi(n/d)=\Omega(\frac{\ell}{n})$. Therefore, each element in $\Z_{n,d}$ is covered $\Omega(\frac{\ell}{n}\phi(n))$ times.
\end{proof}

\subsection{Large divisor with small divisor sum}
\begin{theorem}\label{thm:largel}
If $r = n^{O(\frac{1}{\log \log \log n})}$, then there exists $m|n$, such that $m\geq r$, 
$d(m)=r^{O(\frac{1}{\log \log r})}$ and $\sigma(m) = O(m)$. 
\end{theorem}
\begin{proof}
If there is a single prime $p$, such that $p^e|n$ and $p^e\geq r$, then we pick $m = p^{e'}$, where $e'$ is the smallest integer such that $p^{e'}\geq r$. One can see $d(m) = e' = O(\log r) = r^{O(\frac{1}{\log\log r})}$, also $\varphi(m) = m(1-\frac{1}{p}) \geq \frac{m}{2}$, since $\varphi(m)\sigma(m)\leq m^2$ we are done.

Otherwise, we write $n=\prod_{i=1}^k p_i^{e_i}$, where each $p_i$ is a distinct prime number. The prime $p_i$ are ordered by the weight $w_i=e_ip_i\log p_i$ in decreasing order. That is $w_i\geq w_{i+1}$ for all $i$. Let $j$ be the smallest number such that $\prod_{i=1}^j p_i^{e_i}\geq r$. Let $m=\prod_{i=1}^j p_i^{e_i}$. 

First, we show $d(m)$ is small.
Let $m' = m/p_j^{e_j}$. 
One can see that $m'<r$ and $p_j^{e_j}<r$. So $e_j = O(\log r)$, and
\[
d(m) \leq (e_j+1) d(m') = O(\log r) d(m') = r^{O(\frac{1}{\log \log r})}.
\]

To show that $\sigma(m) = O(m)$, we show $\phi(m) = \Theta(m)$. Indeed, by $\sigma(m)\leq \frac{m^2}{\phi(m)}$, we obtain $\sigma(m)=O(m)$. For simplicity, it is easier to work with sum instead of products, so we take logarithm of everything and define $t=\log n$. By definition, $\log r = O(\frac{\log n}{\log \log \log n}) = O(\frac{t}{\log \log t})$ and $\sum_{i=1}^k e_i \log p_i = t$. 

Note $j$ is the smallest number such that $\sum_{i=1}^j e_i \log p_i \geq \log r$. Because there is no prime $p$ such that $p^e|n$ and $p^e\geq r$, we also have $\sum_{i=1}^j e_i \log p_i< 2\log r = O(\frac{t}{\log \log t})$.

Now, consider $e'_1,\ldots,e'_k$, such that the following holds.

\begin{itemize}
\item $\sum_{i=1}^j e_i\log p_i = \sum_{i=1}^j e_i' \log p_i$, and $e_i'p_i \log p_i = c_1$ for some $c_1$, when $1\leq i \leq j$,
\item $\sum_{i=j+1}^k e_i\log p_i = \sum_{i=j+1}^n e_i' \log p_i$, and $e_i'p_i \log p_i = c_2$ for some $c_2$, where $j+1\leq i \leq k$.
\end{itemize}

Note $c_1$ and $c_2$ can be interpreted as weighted averages over $w_i$. Indeed, consider sequences $x_1,\ldots,x_n$ and $y_1,\ldots,y_n$, such that $\sum_{i}x_i = \sum_{i}y_i$. If for some non-negative $a_1,\ldots,a_n$, we have $a_iy_i=c$ for all $i,j$, then $c \leq \max_{i}a_ix_i$. Indeed, there exists $x_j\geq y_j$, so $\max_{i}a_ix_i \geq a_jx_j\geq a_jy_j=c$. Similarly, $c\geq \min_{i}a_ix_i$. This shows $c_1\geq c_2$, because $c_2\leq \max_{i=j+1}^k w_i = w_{j+1} \leq w_j = \min_{i=1}^j w_i\leq c_1$. 

We first give a lower bound of $c_2$.

$\sum_{i=j+1}^k \frac{c_2}{p_i} = \sum_{i=j+1}^k e'_i\log p_i = \sum_{i=j+1}^k e_i\log p_i \geq t-O(\frac{t}{\log \log t}) = \Omega(t)$.

$\sum_{i=j+1}^k \frac{c_2}{p_i}\leq c_2 \sum_{i=1}^k \frac{1}{p_i}\leq c_2 \sum_{p\text{ prime}, p=O(t)} \frac{1}{p} = c_2 O(\log \log t)$.

This shows $c_2 O(\log \log t) = \Omega(t)$, or $c_2=\Omega(\frac{t}{\log \log t})$.

Since $c_1\geq c_2$,
$\sum_{i=1}^j \frac{1}{p_i} = \sum_{i=1}^j \frac{e_i'\log p_i}{c_1} =\frac{O(\frac{t}{\log \log t})}{c_1} \leq \frac{O(\frac{t}{\log \log t})}{c_2} = \frac{O(\frac{t}{\log \log t})}{\Omega(\frac{t}{\log\log t})} = O(1)$.

Note $\phi(m) = m \prod_{i=1}^j (1-\frac{1}{p_i})$. Because $-2x < \log(1-x) < -x$ for $0\leq x\leq 1/2$, so $\sum_{i=1}^j \log(1-\frac{1}{p_i})\geq -2\sum_{i=1}^j \frac{1}{p_i} = -O(1)$. Hence $\prod_{i=1}^j (1-\frac{1}{p_i}) = \Omega(1)$, and $\phi(m) = \Omega(m)$.
\end{proof}

A interesting number theoretical result is the direct corollary of \Cref{thm:largel}.
\begin{corollary}\label{cor:independentinterest}
Let $n$ be a positive integer, there exists a $m|n$ such that $m = n^{\Omega(\frac{1}{\log \log \log n})}$ and $\sigma(m)=O(m)$.
\end{corollary}


\subsection{Covering of $\D_n$}
Recall that $(\D_n,\odot)$ is the semigroup over the set of divisors of $n$, and the operation $\odot$ is defined as $a\odot b = \gcd(ab,n)$. Throughout this section, we fix a $s\leq n$, and let $A:=\D_n \cap [s]$.
We are interested in finding $s$-coverings of $\D_n$, that is, finding $B\subseteq \D_n$ such that $(\D_n\cap [s])\odot B=\D_n$. 
As we mentioned previously, the main goal is to show that a good $s$-covering of $\D_n$ lifts to a $\ell$-covering of $\Z_n$ of small size. The criteria for a good $s$-covering $B$ is two folds: the size of $B$ should be small ($O(\frac{n}{s} \frac{1}{\log^c n})$), and the reciprocal sum of $B$, namely $\sum_{d\in B}\frac{1}{d}$ should also be small ($O(1)$). However, one can't hope to optimize both at the same time. Fortunately, for our application, we only need the reciprocal sum to be small when $s$ is small.

To obtain a $s$-covering of $\D_n$, there are two natural choices of $B$.
\begin{enumerate}
\item Let $B=(\D_n\setminus [s]) \cup \{1\}$. If $d\leq s$, then $d=d\cdot 1$. Otherwise, if $d > s$, then $d=1\cdot d$. Hence, $A \odot B = \D_n$.
\item Let $B=\D_m$ for some $m|n$ and $m\geq \frac{n}{s}$. We also have $A\odot B=\D_n$. Indeed, consider divisor $d$ of $n$, let $d_1 = \gcd(m,d) \in B$, and $d_2 = d/d_1$. $d_2 | \frac{n}{m} \leq s$, so $d_2\in A$.
\end{enumerate}

These two choices is sufficient for us to prove the following lemma. The lemma basically states there is an $s$-covering of $\D_n$ fits our requirement as long as $s$ is not too large. 

\begin{lemma}\label{lem:B}
Let $\delta$ be a function such that $\delta(n)=\Omega(\log n)$ and $\delta(n)=O(\log^{c'} n)$ for some constant $c'$. 
There exists a constant $c$, such that for every $s\leq \frac{n}{\delta(n)}$, we can find $B\subset \D_n$ such that $(\D_n \cap [s]) \odot B = \D_n$, $|B| = O(\frac{n\log n}{s \delta(n)})$ and 
\begin{enumerate}
\item If $s\in (0,n^{\frac{c}{\log \log n}}]$, then $\sum_{d\in B}\frac{1}{d} = O(\log \log n)$.
\item If $s\in (n^{\frac{c}{\log \log n}},\frac{n}{\delta(n)}]$, then $\sum_{d\in B}\frac{1}{d} = O(1)$.
\end{enumerate}
\end{lemma}

See the proof in \Cref{sec:lembproof}.


\section{$\ell$-covering}\label{sec:bounds}
In this section, we prove our bounds in $f(n,\ell)$ and provide a quick randomized construction. 

\subsection{Upper bound}

The high-level idea is to divide the problem into sub-problems of covering multiple $\Z_{n,d}$. Can we cover $\Z_{n,d}$ for many distinct $d$, using only a few segments in $\mathcal{S}_\ell(\Z^*_n)$? We affirmatively answer this question by connecting an $s$-covering of $\D_n$ to an $\ell$-covering of $\Z_n$. For the remainder of this section, we define $s = \max\left(1,\frac{\ell}{c\log^5 n}\right)$, where $c$ is the constant present in \Cref{thm:betterestimate}. Let $B \subseteq \D_n$ be any $s$-covering of $\D_n$. For each $b\in B$, we generate a cover of all $\bigcup_{d\leq s} \Z_{n,b \odot d}$ using $\mathcal{S}_\ell(\Z_{n,b})$. We denote $g(n,\ell)$ as the size of the smallest set cover of $\bigcup_{d|n,d\leq s} \Z_{n,d}$ using $\mathcal{S}_\ell(\Z^*_n)$. We obtain that
\[
f(n,\ell) \leq \sum_{b\in B} g(\frac{n}{b},\ell).
\]

We provide a bound for $g(n,\ell)$, leveraging the fact that each element is covered multiple times, and \Cref{thm:bettergreedy}, which is the upper bound from the combinatorial set cover theorem.

\begin{theorem}\label{thm:g}
There exists a constant $c>0$, such that
\[g(n,\ell) = \begin{cases}
              O(\frac{n}{\ell}\log \ell) & \text{ if } \ell\geq c \log^5 n,\\
              O(\frac{\phi(n)}{\ell}\log^2 \ell) & \text{ if } c\log^5 n>\ell\geq c \log n.
              \end{cases}\]
\end{theorem}
\begin{proof}
By \Cref{thm:count}, The number of times an element in $\Z_{n,d}$ get covered by a segment in $\mathcal{S}_\ell(\Z^*_n)$ is $\frac{\phi(\frac{n}{d},\floor{\frac{\ell}{d}})}{\phi(\frac{n}{d})} \phi(n)$. We consider $2$ cases. 

Case 1. $\ell>c\log^5 n$. Consider a $d|n$ and $d\leq s$. Then $\lfloor \frac{\ell}{d} \rfloor = \Omega(\log^5 n)$. Hence, $\phi(\frac{n}{d},\lfloor \frac{\ell}{d} \rfloor) = \Omega(\frac{\floor{\frac{\ell}{d}}}{\frac{n}{d}}\phi(\frac{n}{d})) = \Omega(\frac{\ell}{n} \phi(\frac{n}{d}))$ by \Cref{thm:betterestimate}. 
Therefore, each element in $\Z_{n,d}$ is covered by $\frac{\phi(\frac{n}{d},\floor{\frac{\ell}{d}})}{\phi(\frac{n}{d})} \phi(n) = \Omega(\frac{\ell}{n}\phi(n))$ segments in $\mathcal{S}_\ell(\Z^*_n)$. This is true for all element in $\bigcup_{d|n,d\leq s} \Z_{n,d}$.

By \Cref{thm:bettergreedy}, there exists a cover of size 

\[
g(n,\ell) = O\left(\frac{\phi(n)\log \ell}{\frac{\ell}{n}\phi(n)}\right) = O\left(\frac{n}{\ell}\log \ell\right).
\]

Case 2. If $c \log^5 n>\ell\geq c\log n$, then $s=1$, and we try to cover $\Z^*_n$ with $\mathcal{S}_\ell(\Z^*_n)$. Each element is covered by $\frac{\phi(n,\ell)}{\phi(n)}\phi(n) = \Omega(\frac{\ell}{\log \ell})$ segments. By \Cref{thm:bettergreedy}, we have  
\[
g(n,\ell) = O\left(\frac{\phi(n)\log \ell}{\frac{\ell}{\log \ell}}\right) = O\left(\frac{\phi(n)}{\ell}\log^2 \ell\right).\]
\end{proof}

We are ready to prove our main theorem.

\begin{theorem}[Main]\label{thm:main}
There exists an $\ell$-covering set of size $O(\frac{n}{\ell}\log n)$ for all $n, \ell$ where $\ell<n$.
\end{theorem}
\begin{proof}
Let $B$ be the $s$-covering of $\D_n$ in \Cref{lem:B} with $\delta(n) = c\log^5 n$. Observe $s=\frac{\ell}{\delta(n)}$ and $|B| = O(\frac{n}{\ell}\log n)$. 

\paragraph*{Case 1}
If $\ell<c\log n$, then we are done, since $f(n,\ell)\leq n = O(\frac{n}{\ell}\log n)$.

\paragraph*{Case 2} Consider $c\log n\leq \ell \leq c\log^5 n$.
\[
\begin{aligned}
f(n,\ell) &\leq \sum_{d\in B}g(\frac{n}{d},\ell)\\
          &\leq \sum_{d\in B} \left(\phi(n/d) \frac{(\log \ell)^2}{\ell} + 1\right)\\
          &\leq O(\frac{n}{\ell} \log^2 \ell) + |B|\\
          &= O\left(\frac{n}{\ell} (\log \log n)^2\right) + O\left(\frac{n}{\ell}\log n\right)\\
          &=O\left(\frac{n}{\ell}\log n\right)
\end{aligned}
\]

\paragraph*{Case 3}
Consider $\ell > c\log^5 n$.

\[
\begin{aligned}
f(n,\ell) &\leq \sum_{d\in B}g(\frac{n}{d},\ell)\\
          &\leq \sum_{d\in B} O\left(\frac{n}{d} \frac{\log \ell}{\ell}\right) + 1\\
          &=|B|+ O\left(\frac{n\log\ell}{\ell}\right)\sum_{d\in B} \frac{1}{d}\\
          &=O\left(\frac{n}{\ell}\log n\right) + O\left(\frac{n\log\ell}{\ell}\right)\sum_{d\in B} \frac{1}{d}\\
\end{aligned}
\]

Hence, we are concerned with the last term. We further separate into 2 cases:

\paragraph*{Case 3.1} If $\ell < n^{\frac{c}{\log \log n}}$, then $\sum_{d\in B}\frac{1}{d} = O(\log \log n)$, and

\[
\begin{aligned}
O\left(\frac{n\log\ell}{\ell}\sum_{d\in B}\frac{1}{d}\right) &= O\left(\frac{n\log\ell}{\ell}\log \log n\right)\\
          &= O\left(\frac{n\frac{\log n}{\log \log n}\log \log n}{\ell}\right)\\
          &= O\left(\frac{n\log n}{\ell}\right).
\end{aligned}
\]

\paragraph*{Case 3.2} $\ell \geq n^{\frac{c}{\log \log n}}$, then $\sum_{d\in B}\frac{1}{d} = O(1)$. Hence,
\[
\begin{aligned}
O\left(\frac{n\log\ell}{\ell}\sum_{d\in B} \frac{1}{d}\right) = O\left(\frac{n\log \ell}{\ell}\right)=O\left(\frac{n\log n}{\ell}\right).
\end{aligned}
\]

In all cases, we obtain an $\ell$-covering of $O(\frac{n\log n}{\ell})$ size.
\end{proof}

The derived upper bound naturally gives rise to a construction algorithm. Firstly, we find the prime factorization in $n^{o(1)}$ time, and then compute the desired $B$ in $n^{o(1)}$ time. Subsequently, we cover each $\bigcup_{d|n/b, d\leq s}\Z_{n/b,d}$ using $\mathcal{S}_\ell(\Z_{n/b}^*)$ for each $b\in B$. If we apply the linear time greedy algorithm for set cover, then the running time becomes $O(n\ell)$ \cite{Koili.X2019}.

A randomized constructive variant of \Cref{thm:bettergreedy} can also be employed.

\begin{theorem}\label{thm:randomizedrounding}
Let there be $t$ sets, each element of the size $n$ universe is covered by at least $b$ of the sets, then there exists subset of $O(\frac{t}{b}\log n)$ size that covers the universe, and can be found with high probability using a Monte Carlo algorithm that runs in $\tilde{O}(\frac{t}{b})$ time.
\end{theorem}
\begin{proof}[Sketch]
The condition demonstrates that the standard linear programming relaxation of set cover provides a feasible solution, where every indicator variable for each set holds the value of $\frac{1}{b}$. The conventional randomized rounding algorithm, which independently selects each set with a probability equal to $\frac{1}{b}$ for $\Theta(\log n)$ rounds, will cover the universe with high probability \cite{Vazir.2001}. This can be simulated by independently sampling sets of size $\frac{t}{b}$ for $\Theta(\log n)$ rounds, a process that can be completed in $\tilde{O}(\frac{t}{b})$ time.
\end{proof}

The main discrepancy between \Cref{thm:randomizedrounding} and \Cref{thm:bettergreedy} lies in the coverage size. Let $a$ represent the maximum size of each set, the randomized algorithm has a higher factor of $\log n$ rather than $\log a$. If we incorporate more sophisticated rounding techniques, we can once again attain $\log a$ \cite{Srini.1999}. However, the algorithm will slow down. The change from $\log a$ to $\log n$ has implications for the output size. Specifically, following the proof of \Cref{thm:main}, there will be an additional $\log \log n$ factor in the size of the cover.

The analysis mirrors the previous one, enabling us to derive the following theorem.

\begin{theorem}\label{thm:randconstruction}
There exists a constant $c$, such that in $\tilde{O}(\frac{n}{\ell})+n^{o(1)}$ time with high probability, a $\ell$-covering $B$ of $\Z_n$ can be found, such that
\begin{enumerate}
  \item  $|B| = O(\frac{n}{\ell}\log n)$ if $\ell < n^{\frac{c}{\log \log n}}$,
  \item  $|B| = O(\frac{n}{\ell}\log n\log \log n)$ otherwise.
\end{enumerate}
\end{theorem}

\subsection{Lower bound}
We note that our upper bound is optimal through the combinatorial set covering property (\Cref{thm:bettergreedy}). The $\log n$ factor cannot be avoided when $\ell = n^{\Omega(1)}$. To obtain a superior bound, stronger \emph{number theoretical properties} must be leveraged, as was the case when $n$ is a prime \cite{Chen.S.W2013}.

We demonstrate that it is improbable to acquire significantly stronger bounds when \emph{$\ell$ is small}. For an infinite number of $(n,\ell)$ pairs, our bound is merely a $\log \log n$ factor away from the lower bound.

\begin{theorem}
There exists a constant $c>0$, for which there are an infinite number of $n,\ell$ pairs where
$f(n,\ell) \geq c \frac{n}{\ell} \frac{\log n}{\log \log n}$.
\end{theorem}
\begin{proof}
Let $n$ be the product of the smallest $k$ prime numbers, then $k=\Theta(\frac{\log n}{\log \log n})$. Let $\ell$ be the smallest number where $\pi(\ell) = k$. Given that $\pi(\ell) = \Theta(\frac{\ell}{\log \ell})$, we know that $\ell = \Theta(\log n)$.

Note that $\phi(n,\ell)=1$. Indeed, every number $\leq \ell$ except $1$ has a common factor with $n$. To cover all elements in $\Z^*_n\subset \Z_n$, the $\ell$-covering size must be at least $\frac{\phi(n)}{\phi(n,\ell)} = \phi(n) = \Omega(\frac{n}{\log \log n}) = \Omega(\frac{n}{\ell} \frac{\log n}{\log \log n})$.
\end{proof}

\subsection{Application: Simplifying modular subset sum computation}
We demonstrate how our improved bound of $\ell$-covering can be advantageous in algorithm design. $\ell$-covering offers a natural divide-and-conquer algorithm; by partitioning elements into segments in the $\ell$-covering, solving the subproblem, and then combining them together. Such an approach was employed in modular subset sum computations. The modular subset sum problem is defined as follows: Given $S\subset \Z_n$ with $|S|=m$, output all values $i$ such that $\sum_{x\in T} x = i$ for some $T\subset S$.

To solve the modular subset sum, the following theorem was established:
\begin{theorem}[{{\cite[Lemma 5.2]{Koili.X2019}}}]\label{thm:sss}
Let $S\subset \Z_n$ be a set of size $m$, and it can be covered by $k$ segments of length $\ell$, then the subset sums of $S$ can be computed in $O(kn\log n + m\ell \log (m\ell) \log m)$ time.
\end{theorem}

Utilizing the previous $\ell$-covering bound of $O(\frac{n^{1+o(1)}}{\ell})$, a direct application would lead to an $O(\sqrt{m}n^{1+o(1)})$ time algorithm. Instead, in \cite{Koili.X2019}, using a much more intricate recursive partitioning, coupled with a second-level application of \Cref{thm:sss}, Koiliaris and Xu obtained an $O(\sqrt{m}n\log^2 n)$ time algorithm.

Armed with our improved bound on $\ell$-covering, we know $k=O(\frac{n}{\ell}\log n)$. Therefore, setting $\ell = \frac{n}{\sqrt{m}}$, we directly obtain a running time of $O(\sqrt{m}n\log^2 n)$ from \Cref{thm:sss}, matching the significantly more complicated algorithm.

It's worth noting that $\tilde{O}(n)$ time algorithms that completely avoid $\ell$-covering have been discovered \cite{doi:10.1137/1.9781611974782.69,jin_et_al:OASIcs:2018:10043,10.5555/3310435.3310439,doi:10.1137/1.9781611976496.6,potepa:LIPIcs.ESA.2021.76}. However, we continue to believe that $\ell$-covering can provide advantages in future algorithmic applications.
\bibliography{cyclic_cover}

\appendix
\section{Appendix}
\subsection{Proof of \Cref{thm:count}}\label{sec:nproof}

We first show a simple lemma. 
\begin{lemma}\label{lem:tau}
Let $y\in\Z^*_n$, and $B\subset \Z_{n}^*$.
The number of $x\in \Z_{dn}^*$ such that $xb\equiv y \pmod n$, and $b\in B$ is $|B|\frac{\phi(dn)}{\phi(n)}$.
\end{lemma}

\begin{proof}
Indeed, the theorem is equivalent to finding the number of solutions to $x\equiv yb^{-1} \pmod n$ where $b\in B$. For a fixed $b$, let $z=yb^{-1}$. We are asking for the number of $x\in \Z^*_{dn}$ such that $x\equiv z \pmod n$.
Consider the set $A=\set{z+kn \mid 0\leq k\leq d-1}$. Let $P_{n}$ be the set of distinct prime factors of $n$. Since $\gcd(z,n)=1$, no element in $P_n$ can divide any element in $A$. Let $P_{dn}\setminus P_{n}=P_d'\subseteq P_{d}$. Let $q$ be the product of some elements in $P_d'$, $q|d$, $(q,n)=1$. Let $A_q = \set{a \mid a\in A, q|a}$. Note that $q|z+kn \Leftrightarrow k\equiv -zn^{-1} \pmod q$, and given $0\leq k\leq d-1$ and $q|d$, it follows that $|A_q|=\frac{d}{q}$.\\
We can use the principle of inclusion-exclusion to count the elements $a\in A$ such that $\gcd(a,dn)=1$:
\[\sum_{i=0}^{|P_d'|}(-1)^{i}\sum_{S\subseteq P_d',|S|=i}|A_{\prod_{p\in S}p}|=\sum_{i=0}^{|P_d'|}(-1)^{i}\sum_{S\subseteq P_d',|S|=i}\frac{d}{\prod_{p\in S}p}=d\prod_{p\in P_d'}(1-\frac{1}{p})=\frac{\varphi(dn)}{\varphi(n)}.\]
Since all the solution sets of $x$ for different $b\in B$ are disjoint, we find that the total number of solutions over all $B$ is $|B|\frac{\phi(dn)}{\phi(n)}$.
\end{proof}

Now we are ready to prove the theorem.
Since $x\in \Z_n^*$, we observe that $xb\equiv y \pmod n$ if and only if $d|b$, $x\frac{b}{d}\equiv \frac{y}{d} \pmod{\frac{n}{d}}$, and $\frac{b}{d} \leq \floor{\frac{\ell}{d}}$. We can then apply \Cref{lem:tau} and obtain that the number of solutions is $\phi(n/d,\floor{\ell/d})\phi(n)/\phi(n/d)$.

\subsection{Proof of \Cref{lem:B}}\label{sec:lembproof}

\begin{proof}

Let $A=\D_n \cap [s]$. We let $B_1=(\D_n \setminus [s])\cup \{1\}$. Also, let $B_2=\D_m$, where $m|n$, $d(m) = \frac{n}{s}^{O(\frac{1}{\log \log \frac{n}{s}})}$, $\sigma(m)=O(m)$. Such $m$ exists when $s = n^{1-O(\frac{1}{\log \log \log n})}$ by setting $r=\frac{n}{s}$ in \Cref{thm:largel}. Recall both $A\odot B_1=\D_n$ and $A\odot B_2=\D_n$.

The proof consists of $3$ different cases. 
\begin{enumerate}
    \item $s\in (0,n^{\frac{c}{\log \log n}}]$.
    \item $s\in (n^{\frac{c}{\log \log n}},n^{1-\frac{c}{\log \log n}}]$
    \item $s\in (n^{1-\frac{c}{\log \log n}},\frac{n}{f(n)}]$
\end{enumerate}
For the first two cases, we let $B=B_1$.

In particular, we have $s\leq n^{1-\frac{c}{\log \log n}}$, so $\frac{n\log n}{sf(n)} =O(n^\frac{c-\epsilon}{\log \log n})$ for any $\epsilon>0$. Now if we pick sufficiently large $c$, we would have $|B|=d(n) = n^{O(\frac{1}{\log \log n})} = O(\frac{n\log n}{sf(n)})$. 

When $s\in (0,n^{\frac{c}{\log \log n}}]$, $\sum_{d\in B}\frac{1}{d} \leq \frac{1}{n} \sum_{d|n} \frac{n}{d}= \sigma(n)/n = O(\log\log n)$. Otherwise, when $s\in (n^{\frac{c}{\log \log n}},n^{1-\frac{c}{\log \log n}}]$, each element in $B\setminus \{1\}$ is at least $s$, so we know that $\sum_{d\in B} \frac{1}{d} = 1+\sum_{d\in B\setminus \{1\}} \frac{1}{d} \leq 1+|B|\frac{1}{s}\leq 1+\frac{n^{\frac{O(1)}{\log \log n}}}{n^{\frac{c}{\log \log n}}} = O(1)$.

Now, we consider the third case $s\in (n^{1-\frac{c}{\log \log n}},\frac{n}{f(n)}]$. In this case we set $B=B_2$.

We first bound the size of $B$. 
\[
\begin{aligned}
|B| &= (\frac{n}{s})^{O(\frac{1}{\log \log \frac{n}{s}})}\\
          &\leq (\frac{nf(n)}{sf(n)})^{O(\frac{1}{\log \log f(n)})}\\
          &\leq O(\frac{n}{sf(n)}) f(n)^{O(\frac{1}{\log \log f(n)})} \\
          &\leq \frac{n}{sf(n)}(\log^{c'} n)^{O(\frac{1}{\log \log \log n})}\\
          &= O(\frac{n\log n}{sf(n)})\\
\end{aligned}
\]

By the choice of $m$, we have $\sum_{d\in B} \frac{1}{d} = \frac{\sigma(m)}{m} = O(1)$.
\end{proof}

\subsection{Proof of \Cref{thm:betterestimate}}\label{sec:betterestimate}

We first state Brun's sieve.

\begin{theorem}[Brun's sieve \textup{\cite[p.93]{cojocaru2005introduction}} ]\label{thm:sieve}
    
 Let \(\mathcal{A}\) be any set of natural number \(\leq x\) (i.e.
 \(\cal{A}\) is a finite set) and let \(\mathcal{P}\) be a set of primes. For each prime \(p\in\mathcal{P}\), Let \(\mathcal{A}_p\) be the set of elements of \(\mathcal{A}\) which are divisible by \(p\). Let \(\mathcal{A}_1:=A\) and for any squarefree positive integer \(d\) composed of primes of \(\mathcal{P}\) let \(\mathcal A_d:=\cap_{p|d}A_p\). Let \(z\) be a positive real number and let \(P(z):=\prod_{p\in\mathcal{P},p<z}p\).\\
 We assume that there exist a multiplicative function \(\gamma(\cdot)\) such that, for any \(d\) as above,
 \[|\mathcal{A}_d|=\frac{\gamma(d)}{d}X+R_d\]
 for some \(R_d\), where $X:=|A|.$
 We set \[S(\mathcal{A},\mathcal{P},z):=|\mathcal{A}\backslash\cup_{p|P(z)}\mathcal{A}_p|=|\{a:a\in\mathcal{A},\gcd(a,P(z))=1\}|\] and 
 \[W(z):=\prod_{p|P(z)}(1-\frac{\gamma(p)}{p}).\]
 Supposed that\\
 1.\(|R_d|\leq\gamma(d)\) for any squarefree \(d\) composed of primes of \(\mathcal{P}\);\\
 2.there exists a constant \(A_1\geq1\) such that
 \[0\leq\frac{\gamma(p)}{p}\leq 1-\frac{1}{A_1};\]\\
 3.there exists a constant \(\kappa\geq0\) and \(A_2\geq1\) such that
 \[\sum_{w\leq p<z}\frac{\gamma(p)\log p}{p}\leq\kappa\log\frac{z}{w}+A_2\quad\text{if}\quad 2\leq w \leq z.\]
 4.Let \(b\) be a positive integer and let \(\lambda\) be a real number satisfying
 \[0\leq\lambda e^{1+\lambda}\leq 1.\]
 Then 
 \[\begin{aligned}
    S(\mathcal{A},\mathcal{P},z)\geq &XW(z)\{1-\frac{2\lambda^{2b}e^{2\lambda}}{1-\lambda^2 e^{2+2\lambda}}\exp((2b+2)\frac{c_1}{\lambda\log z})\}\\
    &+O(z^{2b-1+\frac{2.01}{e^{2\lambda/\kappa}-1}}),
 \end{aligned}\]
 where \[c_1:=\frac{A_2}{2}\{1+A_1(\kappa+\frac{A_2}{\log 2})\}.\]\\
\end{theorem}

Next, we proceed with the proof.

\begin{proof}
\noindent \textit{Case 1.} $\ell > c \log^5 n$.

Let \(z\) be a value to be determined later. Let \(n_0 = \prod_{p|n,p< z}p\). Observe that \(\phi(n,\ell)\) and \(\phi(n_0,\ell)\) are close. Indeed, for some $c_1>0$,  
\[
\begin{aligned}
    |\phi(n,\ell)-\phi(n_0,\ell)|&=\abs{\sum_{0\leq{m}\leq{\ell},(m,n_0)=1}1-\sum_{0\leq{m}\leq{\ell},(m,n)=1}1}\\
    &\leq \sum_{1\leq{m}\leq{\ell}:p|n,p\geq z,p|m}1\\
    &\leq \sum_{p|n,p\geq z}\frac{\ell}{p}\\
    &\leq \frac{\ell\omega(n)}{z}\\
    &\leq \frac{c_1\ell\log{n}}{z\log \log n}
\end{aligned}
\]

Now, we want to estimate \(\phi(n_0,\ell)\) using the Brun's sieve. The notations are from the theorem. 
Let $\mathcal{A}=\{1,2,\ldots,\ell\}$, $\mathcal{P}=\{p:p|n\}$, $X=|\mathcal{A}|=\ell$, the multiplicative function $\gamma$, where \(\gamma(p)=1\) if \(p\in\mathcal{P}\) otherwise \(0\). 
\begin{itemize}

    \item \textit{Condition (1).} For any squarefree \(d\) composed of primes of \(\mathcal{P}\), 
    \begin{equation*}
        \begin{aligned}
            |R_d| &=\abs{\floor{\frac{\ell}{p}} -\frac{\ell}{p}} \leq 1 = \gamma(d).
        \end{aligned}
    \end{equation*}
    
    \item \textit{Condition (2).} We choose \(A_1\) = 2, therefore \(0\leq \frac{\gamma(p)}{p}=\frac{1}{p}\leq \frac{1}{2} = 1-\frac{1}{A_1}\).
    
    \item \textit{Condition (3).} Because \(R(x):=\sum_{p<x}\frac{\log p}{p}=\log x + O(1)\) \cite{cojocaru2005introduction}, we have 
    \[\sum_{w\leq p<z}\frac{\gamma(p)\log{p}}{p} \leq \sum_{w\leq p<z}\frac{\log{p}}{p} = R(z)-R(w)=\log{\frac{z}{w}}+O(1).\]
    We choose \(\kappa=1\) and some \(A_2\) large enough to satisfy Condition (3).
    
    \item \textit{Condition (4).} By picking \(b=1,\lambda=\frac{2}{9}\), \(b\) is a positive integer and \(0<\frac{2}{9}e^{11/9}\approx 0.75<1\).
\end{itemize}

We are ready to bound \(\phi(n_0,\ell)\). Brun's sieve shows 
\[
\begin{aligned}
    \phi(n_0,\ell)=S(\mathcal{A},\mathcal{P},z)\geq &\ell \frac{\varphi(n_0)}{n_0}\left(1-\frac{2\lambda^{2b}e^{2\lambda}}{1-\lambda^2 e^{2+2\lambda}}\exp((2b+2)\frac{c_1}{\lambda\log z})\right)
    \\
    &+O(z^{2b-1+\frac{2.01}{e^{2\lambda/\kappa}-1}})\\
    \geq &\ell \frac{\varphi(n_0)}{n_0}\left(1-0.3574719\exp(\frac{18 c_1}{\log z})\right)+O(z^{4.59170})
\end{aligned}
\]

Which means that there exists some positive constant \(c_2\) such that for some small $\e>0$,
\[
\phi(n_0,\ell)\geq \ell \frac{\varphi(n_0)}{n_0}\left(1-\frac{2}{5}\exp(\frac{18c_1}{\log z})\right)-c_2z^{5-\e}.
\]

We choose some constant \(z_0\) such that \(\frac{2}{5}\exp(\frac{18c_1}{\log z_0}) \leq \frac12\), if \(z>z_0\)(we will later make sure \(z>z_0\)), then
\[
\phi(n_0,\ell)\geq \frac12 \ell \frac{\varphi(n_0)}{n_0}-c_2z^{5-\e}.
\]

Note if \(n_1|n_2\), then \(\varphi(n_1)/n_1\geq \varphi(n_2)/n_2\) since \(\varphi(n)/n=\prod_{p|n}(1-1/p)\) and every prime factor of $n_1$ is also the prime factor of \(n_2\). Therefore,
\[
\phi(n_0,\ell)\geq \frac12 \ell \frac{\varphi(n)}{n}-c_2z^{5-\e}.
\]

Recall there exists a \(c_3\) such that \(\frac{\phi(n)}{n}\geq\frac{c_3}{\log\log n}\),

\[
\begin{aligned}
    \phi(n,\ell)&\geq \phi(n_0,\ell)-c_1\frac{\ell\log n}{z\log \log n}\\
    &\geq \frac{1}{2} \ell  \frac{\phi(n)}{n}-c_2z^{5-\e} -c_1\frac{\ell\log n}{z\log \log n}\\
    &= \frac{1}{4}\ell\frac{\phi(n)}{n}+(\frac{1}{8}\ell\frac{\phi(n)}{n}-c_2z^{5-\e}) +( \frac{1}{8}\ell\frac{\phi(n)}{n}-c_1\frac{\ell\log n}{z\log \log n})\\
    &\geq \frac{1}{4}\ell\frac{\phi(n)}{n}+(\frac{c_3}{8}\frac{\ell}{\log\log n}-c_2z^{5-\e}) +( \frac{c_3}{8}\frac{\ell}{\log\log n}-c_1\frac{\ell\log n}{z\log \log n}).
\end{aligned}
\]

By picking $z = \frac{8c_1}{c_3}\log n = C\log n$, we obtain $c_1\frac{\ell\log n}{z\log \log n} \leq \frac{c_3}{8}\frac{\ell}{\log\log n}$. By picking $c=8\frac{c_2}{c_3}C^5$ and
$\ell\geq\frac{8c_2}{c_3}C^5\log^{5-\e}n\log\log n=c\log^{5-\e}n\log\log n$, we obtain
$c z^{5-\e}\leq \frac{\ell}{\log\log n}$.

Recall for the above to be true we require $z>z_0$. Note $z=C\log n$, for $z>z_0$ for sufficiently large $n$. If \(n\) is sufficiently large and \(\ell \geq c\log^5n \geq c\log^{5-\e}n\log\log n\), then \(\phi(n,\ell) \geq \frac{\ell}{4n}\varphi(n)\). 
Thus, for all \(n\) and \(\ell \ \geq c\log^5 n\), \(\phi(n,\ell) = \Omega(\ell\frac{\varphi(n)}{n})\). 

\noindent \textit{Case 2.} $\ell > c\log n$.

Observe that for all $\ell\leq n$, $\varphi(n,\ell)\geq 1+\pi(\ell) - \omega(n)$. This is because the primes no larger than $\ell$ are relatively prime to $n$ if it is not a factor of $n$, and $1$ is also relatively prime to $n$.

We show there exists a constant $c$ such that $\varphi(n,\ell)=\Omega(\frac{\ell}{\log \ell})$ for $\ell\geq c \log n$, by showing $\frac{1}{2}\pi(\ell)\geq \omega(n)$. There exists constant $c_1,c_2$ such that $\pi(\ell) \geq c_1\frac{\ell}{\log \ell}$ and $\omega(n) \leq c_2\frac{\log n}{\log \log n}$. Therefore, we want some $\ell$, such that $\frac{c_1}{2}\frac{\ell}{\log \ell} \geq c_2 \frac{\log n}{\log \log n}$. The desired relation holds as long as $\ell \geq c \log n$ for some sufficiently large $c$. 

The constant $c$ in two parts of the proof might be different, we pick the larger of the two to be the one in the theorem.  
\end{proof}

\end{document}